\documentclass[aps,pra,twocolumn,nofootinbib,floatfix,longbibliography,superscriptaddress]{revtex4-2}

\usepackage{amsmath,amssymb,amsthm,bbm,bm,color,dsfont,float,graphicx,hyperref,makecell,mathrsfs,mathtools,nicefrac,pgfplots,physics,tikz,times,txfonts}
\usepackage{comment}
\usepackage[normalem]{ulem}  
\usepackage{empheq}
\definecolor{equationcolor}{RGB}{222,94,100}
\definecolor{alecolor}{RGB}{238,33,80}

\hypersetup{urlcolor=equationcolor, colorlinks=true,citecolor=equationcolor,linkcolor=equationcolor} 
\pgfplotsset{compat=1.18} 

\DeclareFontFamily{U}{mathb}{\hyphenchar\font45}
\DeclareFontShape{U}{mathb}{m}{n}{
	<-6> mathb5 <6-7> mathb6 <7-8> mathb7
	<8-9> mathb8 <9-10> mathb9
	<10-12> mathb10 <12-> mathb12
}{}
\DeclareSymbolFont{mathb}{U}{mathb}{m}{n}
\DeclareMathSymbol{\ggcurly}{\mathrel}{mathb}{"CF}

\renewcommand{\v}[1]{\ensuremath{\boldsymbol #1}}

\makeatletter
\def\blfootnote{\gdef\@thefnmark{}\@footnotetext}
\makeatother

\theoremstyle{plain}
\newtheorem{thm}{Theorem}

\newtheorem{lem}[thm]{Lemma}

\newtheorem{defn}{Definition}

\def\>{\rangle}
\def\<{\langle}

\newlength\myindent
\begin{document}
\newcommand{\lanl}{Theoretical Division (T-4), Los Alamos National Laboratory, Los Alamos, New Mexico 87545, USA.} 
\author{Tanmoy Biswas}
\email{tanmoy.biswas23@gmail.com}
\affiliation{\lanl}

\title{Information and coherence as resources for work extraction\\ from unknown quantum state and providing quantum advantages}

\date{\today}

\begin{abstract}

The amount of extractable work from a physical system is fundamentally connected to the information available about its state, as illustrated by Maxwell’s demon and the Gibbs paradox. In standard thermodynamic protocols involving system–bath interactions, the maximum work is given by the free-energy difference between the initial state and the corresponding Gibbs state at the bath temperature. This motivates a natural question: does information also limit work extraction in closed quantum systems that does not involve a heat bath and work is obtained through unitary operation generated by a time-dependent Hamiltonian? While ergotropy quantifies the maximum work extractable via unitary operations, it assumes complete knowledge of the quantum state, typically requiring full state tomography. In realistic scenarios, however, only partial information is accessible. In this case, the relevant figure of merit is observational ergotropy, which depends explicitly on the measurement used to probe the system. We show that observational ergotropy decreases under classical post-processing of measurement outcomes, implying that fine grained measurements allow greater work extraction than coarse-grained ones. Moreover, maximizing observational ergotropy over all possible measurements recovers standard ergotropy, which decomposes into incoherent (classical) and coherent (quantum) contributions. Our results demonstrate that coherence in the measurement projectors constitutes the key resource, enabling work extraction beyond the incoherent limit and establishing coherence as the origin of quantum advantage in observational ergotropy extraction.

\end{abstract}
\maketitle

\begin{figure}[t]
    \centering
    \includegraphics[width=8cm]{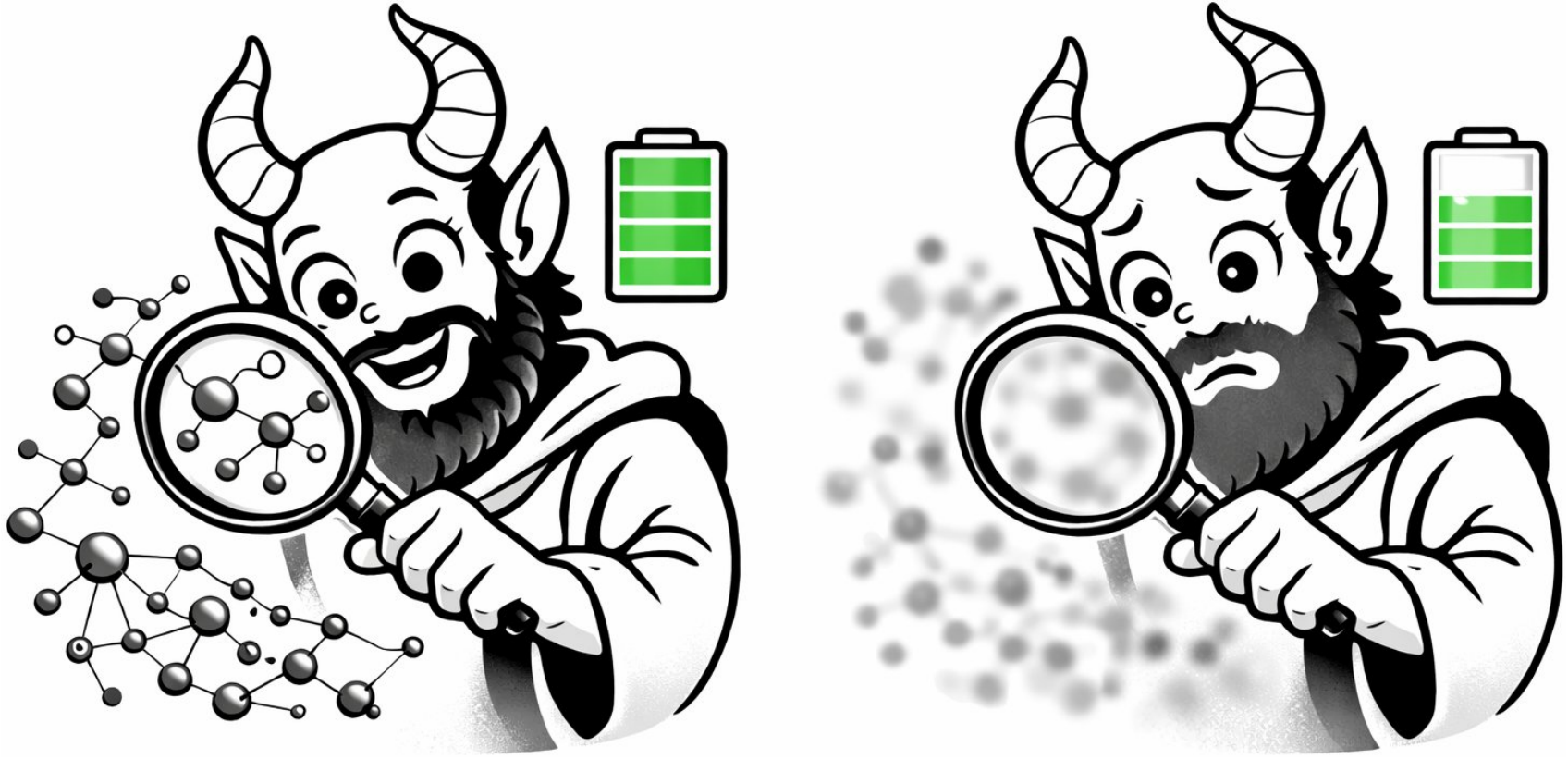}
    \caption{\label{Fig:Maxwell_Demon_glass}
    \textbf{Maxwell’s demon performing sharp and fuzzy measurements for ergotropy extraction.} In the case of a fine-grained (sharp) measurement, the demon acquires more information about the system and is therefore able to extract a larger amount of ergotropy. In contrast, when the same measurement is performed in a fuzzier (coarse-grained) manner, the reduced information available to the demon results in a smaller amount of extractable ergotropy. }
\end{figure}

\section{Introduction}
Information plays a fundamental role in both classical and quantum thermodynamics. In thermodynamics, the ability to extract work is directly tied to the use of available information \cite{JaynesReview,feynman1996computation,bennett1973logical,bennett1982thermodynamics,Parrondo2015}. The profound connection between information and work was first highlighted through the celebrated thought experiment of Maxwell’s demon and later emphasized by Landauer’s principle \cite{landauer1961irreversibility,bennett1987demon,Vedral_Review,Alex_review}. Together, these ideas demonstrate that information is not merely abstract knowledge, but a physical resource with fundamental thermodynamic significance. To acquire relevant information about the underlying state of the thermodynamic system, one must perform measurements on it \cite{bennett1982thermodynamics,Rio2011}. Crucially, the finer and more precise the measurement, the greater the information gained and, consequently, the larger the amount of work that can be extracted \cite{Faist2015,faist2016quantumcoarsegraininginformationtheoreticapproach,Shiraishi_2015}. In contrast, coarse-grained or noisy measurements reduce the available information and thereby limit the amount of extractable work.

The work extraction protocol based on the available information—namely, complete knowledge of the eigenvalues and corresponding eigenvectors of the underlying state—consists of an initial quench followed by an isothermal process performed in contact with a heat bath, thereby rendering the dynamics effectively open-system in nature. In this setting, the optimal protocol allows for the extraction of an amount of work equal to the difference between the free energy of the initial state and that of the corresponding Gibbs state at the temperature of the bath \cite{Rio2011,Szilard1929}. This naturally raises the following question: Does increased information about the underlying state enable greater work extraction, even in a closed-system setting where the system remains isolated from the heat bath? 

In a closed-system setting, the maximum amount of extractable work is characterized by ergotropy of the underlying state \cite{Ergotropy}. For a quantum state $\rho$ associated with a Hamiltonian $H$, the ergotropy is defined as
\begin{equation}\label{defn:Ergotropy}
    R(\rho):=\Tr(H\rho)-\Tr(H\Pi),
\end{equation}
where $\Pi$ denotes the passive state associated with $\rho$, defined through
\begin{equation}\label{passive_state_energy}
    \Tr(H\Pi):=\min_{U\in\mathcal{U}(d)}\Tr(HU\rho U^{\dagger}),
\end{equation}
and $\Tr(H\Pi)$ referred to as \emph{passive energy} of $\rho$ \cite{SilvaPRE1,Llobet,Biswas2022extractionof}. Assuming that the eigenvalues of $H$ are arranged in increasing order (i.e., $E_i \leq E_{i+1}$ for all $i$), it follows that
\begin{equation}\label{pass_energy_major}
    \Tr(H \Pi) = \lambda(H)^{T} \lambda^{\downarrow}(\rho),
\end{equation}
where $\lambda(H)$ is a column vector whose $i^{th}$ element is $E_i$ and $T$ denotes transpose. $\lambda^{\downarrow}(\rho)$ denotes the column vector of eigenvalues of $\rho$ arranged in decreasing order 
(i.e., $\lambda_i(\rho) \geq \lambda_{i+1}(\rho)$ for all $i$). Note that $\Pi$ is obtained from $\rho$ by a unitary transformation $S$ generated by a cyclic driving Hamiltonian $H(t)$,
\begin{eqnarray}
    S = \mathcal{T}\left[\exp\!\left(-i\int_0^\tau dt \, H(t)\right)\right],
\end{eqnarray}
with $H(0)=H(\tau)\equiv H$. Since the change in average energy arises solely from the time-dependent driving, ergotropy admits a natural interpretation as work extracted in a closed system \cite{Alicki_1979,Kosloff1984}. 

The unitary transformation \( S \) that achieves maximal work extraction i.e ergotropy—implicitly assumes complete knowledge of the spectral decomposition of the quantum state \( \rho \), including both its eigenvalues and corresponding eigenvectors. In practice, obtaining such information requires full quantum state tomography, whose experimental cost scales exponentially with system size and is therefore infeasible for large systems. This limitation motivates a more realistic framework in which the state is only partially known, or even entirely unknown. In this setting, the relevant notion of extractable work is \emph{observational ergotropy} introduced in Ref. \cite{SafranekBinderPRL}, defined as the work operationally accessible from the information obtained via a single coarse-grained measurement, without prior knowledge of the prepared state \cite{Wehrl1978,vonNeumann2010,SafranekPRA,Buscemi_2023}. This naturally raises two fundamental questions: whether having greater information about the underlying system always enables increased observational work extraction, and what properties a coarse-graining measurement must possess to provide a genuine quantum advantage in observational ergotropy extraction.


In this Letter, we show that observational ergotropy decreases under classical post-processing of the quantum measurements used to extract information about the underlying physical system. In particular, fine grained projective measurements enable a larger amount of observational ergotropy to be extracted than their classically post-processed (coarse-grained) counterparts, establishing that any loss of information at the level of measurement outcomes necessarily reduces the extractable work. Furthermore, we prove that observational ergotropy, when maximized over all fine-grained measurement and their classical post-processing, reduces to the standard ergotropy defined in Eq.~\eqref{defn:Ergotropy}. Since ergotropy admits a decomposition into incoherent (classical) and coherent contributions, we show that measurements performed in the Hamiltonian eigenbasis yield observational ergotropy equal to the incoherent part alone. This provides an operational characterization of measurements capable of extracting observational ergotropy beyond the classical limit: energetic coherence in the measurement projectors emerges as the essential resource enabling work extraction exceeding the incoherent ergotropy of the state.
\section{Preliminaries}
Before presenting the main results, we first introduce the essential concepts required for their understanding. We proceed by defining fine-grained measurements and coarse-grained measurements \cite{oftelie2025impactinformationquantumheat,Oszmaniec,Jozsa2003EntanglementCost,DAriano_2005}.
\begin{defn}[Fine-grained measurement]\label{Defn_fine_grained_mmt}
    A \emph{fine-grained measurement} is a projective measurement $\textbf{P}=\{P_i\}_{i=1}^d$ in which each projector has rank one, i.e.,
\begin{align}
    P_i P_j &= \delta_{ij} P_i, \qquad \sum_{i=1}^d P_i = \mathbb{I}, \\
    \text{and} \qquad \mathrm{Rank}(P_i) &= 1 \quad \forall\, i \in \{1,\ldots,d\}.
\end{align}
\end{defn}
In order to describe one measuremnt is less informative than the other, we need to introduce \emph{coarse-grained} measurement
\begin{defn}[Coarse-grained measurement]
A measurement $\mathbf{Q}=\{Q_i\}_{i=1}^{n}$ is called a \emph{coarse-grained} measurement if there exists a fine-grained measurement $\mathbf{P}=\{P_j\}_{j=1}^{d}$ and a classical post processing defined via column-stochastic matrix $D$ such that
\begin{equation}\label{classical_postprocessing}
    Q_i = \sum_{j=1}^{d} D_{i,j} P_j, \qquad \forall i ,
\end{equation}
where $D_{i,j} \ge 0$ for all $i,j$ and $\sum_i D_{i,j} = 1$ for every $j$.
\end{defn}
In other words, classical post-processing of a fine-grained measurement is a strategy in which, upon obtaining an outcome $j$, one outputs the final outcome $i$ with probability $D_{i,j}$ \cite{Heinosaari2022orderpreservingmaps}. We note that the coarse-grained measurement $\mathbf{Q}=\{Q_i\}$ defines a Positive Operator-Valued Measure (POVM), since
\begin{equation}
    \forall i,\quad Q_i \ge 0 
    \qquad \text{and} \qquad 
    \sum_i Q_i = \mathbb{I}.
\end{equation}
A special form of classical post-processing arises when the column-stochastic matrix $D$ has entries only in $\{0,1\}$; such matrices correspond to deterministic coarse-graining and are referred to as \emph{relabeling matrices}.

Observe that the number of elements in the coarse-grained measurement $\mathbf{Q}$ need not coincide with the number of elements in the fine-grained measurement $\mathbf{P}$ from which $\mathbf{Q}$ is obtained via classical post-processing. We denote the set of all coarse-grained measurements as $\mathcal{P}(d,n)$. Since any fine-grained measurement can be regarded as a trivial classical post-processing of itself—corresponding to the choice \( D_{i,j} = \delta_{ij} \) in Eq. \eqref{classical_postprocessing} —it follows that every fine-grained measurement is naturally included in \( \mathcal{P}(d,n) \).

Next, to describe an estimate of the initial state $\rho$ obtained from a single coarse-grained measurement, we introduce the notion of a coarse-grained state.

\begin{defn}[Coarse-grained state]
Let $\mathbf{M}=\{M_i\}_{i\in I}$ be a coarse-grained measurement. For any quantum state $\rho$, the corresponding coarse-grained state is defined as
\begin{align}
    \rho_{\mathrm{cg}}^{(\textbf{M})}
    := \sum_{i \in I} p_i \, \frac{M_i}{V_i},
\end{align}
where $p_i = \mathrm{Tr}(\rho M_i)$ denotes the probability of obtaining outcome $i$, and $V_i := \mathrm{Tr}(M_i)$ is the volume of the $i$-th coarse-graining element.
\end{defn}

We conclude with a remark on the interpretation of the coarse-grained state. 
Suppose that $M_i$ is the eigenprojector of the Hamiltonian corresponding to the energy eigenvalue $E_i$. 
Then the normalized operator $M_i/\Tr(M_i)$ represents the microcanonical ensemble, i.e., the state of maximal entropy compatible with the fixed energy value $E_i$. More generally, let $M_i$ denote the eigenprojector of a relevant observable associated with the eigenvalue $m_i$. 
The normalized state
\begin{equation}
    \frac{M_i}{V_i} = \frac{M_i}{\Tr(M_i)}
\end{equation}
corresponds to the state of maximal ignorance (maximum entropy in the information-theoretic sense) consistent with the knowledge that the measurement outcome is $m_i$. This is precisely the generalized microcanonical ensemble associated with $m_i$. 

Accordingly, the coarse-grained state can be interpreted as a statistical mixture of such microcanonical ensembles and is therefore completely characterized by the probabilities $p_i = \Tr(M_i \rho)$ together with the volumes $V_i = \Tr(M_i)$ of the corresponding eigenspaces. 
Furthermore, one can show that $\rho_{\mathrm{cg}}$ admits a natural Bayesian interpretation as the optimal estimate of $\rho$ based solely on coarse-grained information. In particular, it arises from the action of the Petz recovery map~\cite{Petz1,Petz2} as well as from its rotated variant~\cite{Junge2018}. A detailed review about coarse-grained states and its relevance in stochastic and quantum thermodynamics can be found in \cite{StrasbergWinter,SafranekTutorial}.

To characterize maximal work extraction or ergotropy—in situations where no prior knowledge about the underlying quantum state is available beyond the information obtained from a single coarse-grained measurement, we introduce the notion of \emph{observational ergotropy}, as proposed in Ref.~\cite{SafranekBinderPRL}. 
\begin{defn}[Observational Ergotropy]
   For any coarse-grained measurement $\mathbf{M}=\{M_i\}_{i\in I}$, and for any quantum state $\rho$ associated with Hamiltonian $H$, the \emph{observational ergotropy} $R(\rho,\mathbf{M})$ is defined as
\begin{align}\label{Observational_ergotropy_defn_Eq}
    R(\rho,\mathbf{M})
    := \Tr\!\left[H\left(\rho-\Pi_{\mathrm{cg}}^{(\textbf{M})}\right)\right],
\end{align}
where the coarse-grained passive state $\Pi_{\mathrm{cg}}^{(\textbf{M})}$ is defined using the following relation:
\begin{align}
    \min_{U\in\mathcal{U}(d)} \Tr\!\left(HU \rho_{\mathrm{cg}}^{(\textbf{M})} U^{\dagger}\right)
    := \Tr\!\left(H\Pi_{\mathrm{cg}}^{(\textbf{M})}\right).
\end{align}
Here $p_i=\Tr(\rho M_i)$ denotes the probability of outcome $i$, and $V_i:=\Tr(M_i)$ referred to as the volume associated with $M_i$.
\end{defn}

The procedure for extracting observational ergotropy, along with its connection to observational entropy in the asymptotic regime and to entanglement entropy, has been discussed in detail in Refs.~\cite{SafranekBinderPRL,Safranekarxiv}. We now proceed to present the main results.

\section{Main Results: }
\subsection{Role of information in observational ergotopy extraction}
\begin{thm}[Fine-grained measurements always yield higher observational ergotropy than coarse-grained measurements]\label{fmt}
    Consider a fine grained measurement $\v P:=\{P_i\}_{i=1}^d$ and the coarse-grained measurement $\textbf{Q}$ which is formed by classical post-processing of $\v P$ i.e., 
    \begin{equation}\label{Q_i}
        Q_i = \sum_{j}D_{i,j}P_j, \quad\forall i,
    \end{equation}
    where $D_{i,j}$ are elements of a column-stochastic matrix. Then, observational ergotropy satisfies the following inequality:
    \begin{equation}\label{Ineq_Observational_ergotropy_class}
        R(\rho, \textbf{P})\geq R(\rho, \textbf{Q}).
    \end{equation}
\end{thm}
The above theorem shows that fine-grained measurements yield higher observational ergotropy than their classically post-processed (coarse-grained) counterparts. Note that, the inequality in Eq.~\eqref{Ineq_Observational_ergotropy_class} is in the same spirit as the Gibbs paradox: having more information about the underlying thermodynamic system enables one to extract more work than in a scenario where less information is available. The proof of the above theorem can be found in the end matter.

\subsection{Quantum advantage in ergotropy extraction} In the following, we address which types of coarse-grainings can provide a genuine quantum advantage in ergotropy extraction. To formulate this question precisely, we recall an important result from Ref.~\cite{FrancicaPRL}, which shows that the ergotropy of a state $\rho$ admits a decomposition into incoherent (classical) and coherent (quantum) contributions,
\begin{equation}\label{decomposition_ergotropy}
    R(\rho)= R_{\mathrm{incoherent}}(\rho)+ R_{\mathrm{coherent}}(\rho).
\end{equation}
Here,
\begin{align}\label{incoherent_ergotropy}
    R_{\mathrm{incoherent}}(\rho) = \Tr\left(H\Delta_{\rho}\right)-\min_{V\in\mathcal{U}(d)}\Tr\left(HV\Delta_{\rho}V^{\dagger}\right),
\end{align}
where $\Delta_{\rho}$ is dephased version of the state $\rho$ in the energy eigenbasis
\begin{equation}\label{dephased_state}
    \Delta_{\rho} = \sum_{i=1}^d \langle E_i|\rho|E_i\rangle\ketbra{E_i}{E_i},
\end{equation}
whereas $R_{\mathrm{coherent}}(\rho)$ defined using Eq. \eqref{defn:Ergotropy} and Eq. \eqref{decomposition_ergotropy}. From Eq.~\eqref{incoherent_ergotropy}, we observe that extracting incoherent ergotropy requires optimization only over permutation unitaries, since these act on the dephased state. 

Our goal is to characterize the resourceful coarse-grained measurements
for which the observational ergotropy exceeds the incoherent contribution, thereby certifying a genuine quantum advantage. 

To this end, we begin by characterizing those coarse-grained measurements that yield observational ergotropy equal to the incoherent ergotropy. We consider the set of all coarse-grained measurements $\mathcal{P}_{H}(d)$ that can be obtained via classical post-processing of the projective measurement defined by the eigenprojectors of the Hamiltonian. We refer to such measurements as \emph{energy-incoherent measurements}, since their elements are diagonal in the energy eigenbasis and thus contain no coherence with respect to the Hamiltonian \cite{Oszmaniec2019operational,Theurer}. 
More explicitly, any POVM $\mathbf{N}=\{N_i\}_{i=1}^n$ in $\mathcal{P}_{H}(d,n)$ can be written as
\begin{equation}\label{energy_inc_mmt}
    N_i = \sum_{j} q(i|j)\ketbra{E_j}{E_j},
\end{equation}
where $H=\sum_{j=1}^{d} E_j \ketbra{E_j}{E_j}$ is the spectral decomposition of the Hamiltonian and $q(i|j)$ is a column stochastic matrix describing the classical post-processing.
We will show that they constitute the coarse-grained measurements that results observational ergotropy equal to incoherent ergotropy in the following theorem:
\begin{thm}\label{thm_incoherent_ergotropy}
For any measurement $\mathbf{N} \in \mathcal{P}_{H}(d)$, we have
\begin{equation}\label{maximization_ergotropy}
    \max_{\mathbf{N}\in \mathcal{P}_{H}(d)} R(\rho,\mathbf{N})
    =
    R_{\mathrm{incoherent}}(\rho).
\end{equation}
Here, $\mathcal{P}_{H}(d)$ denotes the set of energy incoherent POVMs defined in Eq. \eqref{energy_inc_mmt}. 
\end{thm}
\begin{proof}
Employing Theorem \ref{fmt}, we infer that the maximization in Eq.~\eqref{maximization_ergotropy} is achieved for the projective measurement 
$\textbf{E}=\{\ketbra{E_i}{E_i}\}_{i=1}^{d}$, where $\{\ket{E_i}\}_{i=1}^{d}$ denotes the energy eigenbasis of the Hamiltonian. Then, the observational ergotropy is given by
\begin{equation}
    R(\rho,\textbf{E})=\Tr(H\rho)-\Tr(H\Pi^{(\textbf{E})}_{\text{cg}})
\end{equation}
where $\Pi^{(\textbf{E})}_{\text{cg}}$ is the passive version of the coarse-grained states obtained from the projective measurement $\textbf{E}$,
\begin{equation}
    \rho_{\text{cg}}^{(\textbf{E})} = \sum_{i=1}^d \langle E_i|\rho|E_i\rangle\ketbra{E_i}{E_i}.
\end{equation}
We observe that $\rho_{\text{cg}}^{(\textbf{E})}$ is the dephased version of the state $\rho$ in the energy eigenbasis i.e., $\rho_{\text{cg}}^{(\textbf{E})}=\Delta_{\rho}$ as can be seen from Eq. \eqref{dephased_state}. As a consequence, the average energy satisfies the following equality:
\begin{equation}\label{equality_theorem_inc}
    \Tr(H\rho) = \Tr(H\rho_{\mathrm{cg}}^{(\textbf{E})}).
\end{equation}
This allow us to write the following:
\begin{align}
    R(\rho,\textbf{E})&=\Tr(H\rho)-\Tr(H\Pi^{(\textbf{E})}_{\text{cg}})\nonumber\\ & = \Tr(H\rho)-\Tr(H\rho_{\text{cg}}^{(\textbf{E})})+\Tr(H\rho_{\text{cg}}^{(\textbf{E})})-\Tr(H\Pi^{(\textbf{E})}_{\text{cg}})\nonumber\\
    &= \Tr(H\rho_{\text{cg}}^{(\textbf{E})})-\Tr(H\Pi^{(\textbf{E})}_{\text{cg}})=R_{\text{incoherent}}(\rho),
\end{align}
where to write the third equality, we use Eq. \eqref{equality_theorem_inc}.
\end{proof}

As we have mentioned earlier that coarse-graining in the Hamiltonian eigenbasis simply yields the dephased version of the initial state. Such a procedure is completely insensitive to energetic coherences and, consequently, to the coherent contribution to the ergotropy. 

It follows from theorem \ref{thm_incoherent_ergotropy}, in order to surpass the incoherent (classical) contribution to ergotropy, the coarse-grained measurements must itself possess coherence in the energy eigenbasis \cite{Lostaglio2015,StudziskiHorodecki,KorzekwaPRX,Lobejko2021}. Equivalently, only measurements whose elements are not diagonal in the Hamiltonian eigenbasis can provide a genuine quantum advantage. In this sense the energy-coherent coarse-grainings are resourceful \cite{PlenioPRL,Stretsov_Review}. To end this, we show that maximizing the observational ergotropy over all coarse-grained measurements (This also includes coarse-grainings that contain coherence in the energy eigenbasis) recovers the full ergotropy of the state.

\begin{thm}\label{theorem3}
    For any measurement $M\in\mathcal{P}(d,n)$
    \begin{equation}\label{max1}
        \max_{\textbf{M}\in\mathcal{P}(d,n)} R(\rho, M)=R(\rho),
    \end{equation}
    where $\mathcal{P}(d,n)$ denote the set of all fine-grained measurement and their all possible classical postprocessing.
\end{thm}
\begin{proof}
    We begin by observing that the maximization in Eq.~\eqref{max1} is achieved by a fine-grained measurement, as follows from Theorem~\ref{fmt}. For a given quantum state $\rho$, let $\mathbf{M}^{*}=\{\ketbra{k}{k}\}_{k=1}^{d}$ denote the fine-grained measurement such that:
\begin{equation}\label{max_observational_ergotropy}
    \max_{\mathbf{M}\in\mathcal{P}(d)} R(\rho,\mathbf{M})
    = R(\rho,\mathbf{M}^{*}).
\end{equation}
The coarse-grained state corresponding to $\mathbf{M}^{*}$ is then given by
\begin{equation}\label{rho_cgm}
    \rho_{\mathrm{cg}}^{(\textbf{M}^{*})}
    = \sum_{k=1}^{d} p_k \ketbra{k}{k},
    \qquad
    p_k = \langle k|\rho|k\rangle .
\end{equation}

First, choosing $\mathbf{M}$ to be the eigenbasis of $\rho$, it follows immediately that $R(\rho,\mathbf{M}) = R(\rho)$. Hence, employing Eq. \eqref{max_observational_ergotropy}
\begin{equation}\label{Equality1}
    R(\rho,\mathbf{M}^{*}) \geq R(\rho,\mathbf{M}) = R(\rho).
\end{equation}

We now show that $R(\rho,\mathbf{M}^{*}) \leq R(\rho)$. Let the spectral decomposition of $\rho$ be
\begin{equation}
    \rho = \sum_{i=1}^{d} \alpha_i \ketbra{\alpha_i}{\alpha_i}.
\end{equation}
We first prove that $\lambda(\rho)\succ \lambda(\rho_{\mathrm{cg}}^{(\textbf{M}^{*})})$, where
$\lambda(\rho)=(\alpha_1,\ldots,\alpha_d)^T$ and
$\lambda(\rho_{\mathrm{cg}}^{(\textbf{M}^{*})})=(p_1,\ldots,p_d)^T$ (See section \ref{Review_maj} of the end matter).
Using Eq. \eqref{rho_cgm}, we can write
\begin{align}
    p_k
    &= \langle k|\rho_{\mathrm{cg}}^{(\textbf{M}^{*})}|k\rangle
      = \sum_{i=1}^{d} \alpha_i |\langle k|\alpha_i\rangle|^2 \nonumber \\
    &= \sum_{i=1}^{d} \alpha_i |\langle k|V|i\rangle|^2
      = \sum_{i=1}^{d} D_{k,i}\alpha_i ,
\end{align}
where $V$ is the unitary relating the eigenbasis of $\rho$ to $\{\ket{k}\}$, and
$D_{k,i}:=|\langle k|V|i\rangle|^2$. Since $D_{k,i}\geq 0$ and
\begin{equation}
    \sum_{i=1}^{d} D_{k,i}=1,
    \qquad
    \sum_{k=1}^{d} D_{k,i}=1,
\end{equation}
the matrix $D$ is bistochastic. This implies
$\lambda(\rho)\succ \lambda(\rho_{\mathrm{cg}}^{(\textbf{M}^*)})$.
By using Schur concavity of passive energy proven in Theorem~\ref{SCofOE} of the end matter, we can conclude that the passive energy of
$\rho_{\mathrm{cg}}^{(\textbf{M}^*)}$ is larger than that of $\rho$, i.e.,
\begin{equation}
    \Tr(H\Pi_{\mathrm{cg}}^{(\textbf{M}^*)})
    \leq
    \Tr(H\Pi_{\mathrm{cg}}),
\end{equation}
which yields
\begin{equation}\label{Equality2}
    R(\rho)
    = \Tr(H\rho)-\Tr(H\Pi_{\mathrm{cg}})
    \geq
    \Tr(H\rho)-\Tr(H\Pi_{\mathrm{cg}}^{(\textbf{M}^*)})
    = R(\rho,\mathbf{M}^{*}).
\end{equation}

Combining Eqs.~\eqref{Equality1} and~\eqref{Equality2} completes the proof.
\end{proof}
Therefore, from Theorem~\ref{theorem3}, we see that coherence associated with the projectors of a coarse-grained measurement in the energy eigenbasis is crucial, as it enters the construction of the coarse-grained states and consequently enables the observational ergotropy to exceed the incoherent ergotropy, thereby providing a genuine quantum advantage.

\section{Discussion}
In conclusion, we develop an operational framework that clarifies how accessible information fundamentally constrains work extraction in closed quantum systems. We show that observational ergotropy—quantifying the maximum extractable work when the underlying state is partially or completely unknown—monotonically decreases under classical post-processing. Thus, any information loss induced by coarse-graining inevitably reduces the extractable work. Our results therefore place information at the heart of work extraction in closed systems, where no heat bath is involved. In this sense, our findings echo the reformulation of the Gibbs paradox by E.~T.~Jaynes, who emphasized that \emph{``the amount of useful work that we can extract from any system depends---obviously and necessarily---on how much `subjective' information we have about its microstate, because that tells us which interactions will extract energy and which will not''} \cite{Jaynes1992}. Here, however, the work extraction protocol operates in the closed system regime i.e., entirely without coupling to a thermal reservoir, highlighting that the informational origin of thermodynamic performance persists even in fully closed quantum settings.

Moreover, we show that optimizing observational ergotropy over all coarse-grained measurements recovers the standard notion of ergotropy, thereby establishing a direct connection between observational and standard ergotropy. By decomposing ergotropy into incoherent and coherent contributions, we precisely identify the origin of quantum advantages in observational ergotropy extraction: only measurements that possess coherence in the energy eigenbasis can access the coherent part of ergotropy. This reveals that coherence in the measurement operators constitutes a genuine thermodynamic resource, responsible for enabling quantum advantages in observational ergotropy extraction. Overall, our results deepen the interplay between information theory, quantum coherence, and thermodynamics, and offer a clear operational perspective on how measurement resources fundamentally determine thermodynamic performance in quantum systems.

Our work opens several avenues for future research. In particular, it would be interesting to investigate how observational ergotropy applies to many-body quantum batteries, where the energy levels are often so closely spaced that full ergotropy extraction becomes impractical due to the requirement of extremely precise control \cite{Gyhm,RossiniPRL,GhoshChandaPRA,CampaioliRev}. In such settings, observational ergotropy provides a more realistic and operationally meaningful figure of merit. It would therefore be valuable to characterize which classes of coarse-grained measurements are most effective for specific quantum battery models, and to understand how the battery capacity of such models is influenced by the coherence present in the projectors defining the coarse-grained measurement. Another important research direction is to quantify the energetic cost of performing such measurements and to understand how this cost depends on their sharpness. This would enable us to establish a clear trade-off between the energy invested in the measurement process and the amount of ergotropy that can be extracted, thereby providing a more complete thermodynamic assessment of measurement-assisted work extraction \cite{xuereb2024resources,abdelkhalek2018fundamentalenergycostquantum,Guryanova2020idealprojective}. It would also be of considerable interest to investigate the roles and potential quantum advantages of other non-classical resources—such as entanglement and magic—in the extraction of observational ergotropy \cite{Francica2017,BiswasDattaPRL,biswas2025steerablequantumcorrelationsprovide,junior2025tradingathermalitynonstabiliserness,HsiehPRL}. Finally, it will be very interesting to explore the roles  of observational ergotropy in heat engines \cite{LutzPRL,PandaPRA,ElouradJordan,BiswasPRL,BiswasPRE}.

\section*{Acknowledgment} T.B acknowledges Luis Pedro Garcia-Pintos, A. de Oliveira Junior and Dominik Safranek for insightful discussions and comments.  This work was supported by the U.S. Department of Energy, Office of Science, Advanced Scientific Computing Research program under project TEQA.
\bibliography{literature}
\newpage

\section*{End Matter}
\subsection{Brief overview of majorization theory}\label{Review_maj}
In this section, we briefly review majorization theory required for understanding the technical proofs. A detailed discussion of majorization theory can be found in Ref. \cite{MarshallOlkinArnold2011,Bhatia1997,Gour_2025,GOUR20151}.
\begin{defn}[Majorization]
Let $x,y \in \mathbb{R}^n$, and let $x^{\downarrow}$ and $y^{\downarrow}$ denote the vectors obtained by rearranging their components in non-increasing order.  
We say that $x$ is \emph{majorized} by $y$, written $x \prec y$, if
\begin{align}
\sum_{i=1}^{k} x_i^{\downarrow}
\;\le\;
\sum_{i=1}^{k} y_i^{\downarrow}
\quad \text{for all } k=1,\dots,n-1,
\end{align}
and
\begin{align}
\sum_{i=1}^{n} x_i^{\downarrow}
=
\sum_{i=1}^{n} y_i^{\downarrow}.
\end{align}
\end{defn}

Next, we state a set of equivalent conditions that characterize the majorization relation $x \prec y$ in the following theorem.
\begin{thm}\label{thm_majorization_equvt}
For $x,y \in \mathbb{R}^n$, the following are equivalent:

\begin{enumerate}
    \item $x \prec y$.

    \item There exists a bistochastic matrix $B$ such that
    \begin{align}
        & x = B y \quad\text{where}\quad \forall \;i,j\;\; B_{i,j}\geq 0 \nonumber\\&\text{and}\quad \sum_{i}B_{i,j}=\sum_{j}B_{i,j}=1.
    \end{align}

    \item $x$ lies in the convex hull of permutations of $y$, i.e.,
    \begin{align}
        x = \sum_{j} \pi_j P_j y,
    \end{align}
    where $\{p_j\}$ is a probability distribution and $\pi_j$ are permutation matrices.

    \item For every Schur-concave function $f$,
    \begin{align}
        x \prec y \;\Rightarrow\; f(x) \ge f(y),
    \end{align}
    whereas for every Schur-convex function $g$
    \begin{align}
        x \prec y \;\Rightarrow\; g(x) \le g(y).
    \end{align}
\end{enumerate}
\end{thm}
\subsection{Proof of theorem \ref{fmt}}
To prove Theorem \ref{fmt}, we proceed by proving two important results as follows:
\begin{lem}\label{fmt_lem}
    Consider a fine-grained measurement $\v P:=\{P_i\}_{i=1}^d$ and the coarse-grained measurement $\textbf{Q}$ which is formed by classical post-processing of $\v P$ i.e., 
    \begin{equation}\label{Q_i}
        Q_i = \sum_{j}D_{i,j}P_j, \quad\forall i,
    \end{equation}
    where $D_{i,j}$ are elements of a stochastic matrix.
    Then the spectrum of the coarse-grained state $\rho_{\text{cg}}^{(\textbf{P})}$ and $\rho_{\text{cg}}^{(\textbf{Q})}$ satisfy the following majorizaton condition:
    \begin{equation}
        \lambda(\rho_{\text{cg}}^{(\textbf{P})})\succ \lambda(\rho_{\text{cg}}^{(\textbf{Q})}).
    \end{equation}
\end{lem}
\begin{proof}
We begin by writing the coarse-grained states arising from $\v P$
\begin{equation}\label{cg_P}
    \rho^{(\textbf{P})}_{cg}= \sum_{j}\Tr(\rho P_j)\frac{P_j}{\Tr(P_j)} := \sum_{j}p_j\frac{P_j}{\Tr(P_j)},
\end{equation}
where $p_j=\Tr(\rho P_j)$.
For the coarse-grained measurement $\v Q$ obtained from classical post-processing of $\v P$, we have  
    \begin{align}\label{qji}
        \forall i,\quad\text{we have}\quad\frac{Q_i}{\Tr(Q_i)} &= \frac{\sum_{j}D_{i,j}P_j}{\Tr\left(\sum_{j}D_{i,j}P_j\right)} \nonumber\\&= \sum_{j}\frac{D_{i,j}\Tr(P_j)}{\sum_{k}D_{i,k}\Tr\left(P_k\right)}\times\frac{P_j}{\Tr(P_j)}\nonumber\\&=\sum_{j}q(j|i)\frac{P_j}{\Tr(P_j)},
    \end{align}
    where we defined
    \begin{equation}\label{q_conditional_i}
        q(j|i) := \frac{D_{i,j}\Tr(P_j)}{\sum_{k}D_{i,k}\Tr\left(P_k\right)},
    \end{equation}
   and $D$ is a column stochastic matrix. Note that $\forall i$, the quantities $q(j|i)$ given in Eq. \eqref{qji} describes a valid probability vector as  
    \begin{equation}
        \forall i,\;j \quad q(j|i) \geq 0\quad\text{and}\quad \sum_{j}q(j|i) =1.
    \end{equation}
    
    Therefore, the coarse-grained state $\rho^{(\textbf{Q})}_{cg}$, we have
    \begin{align}\label{cg_Q}
        \rho^{(\textbf{Q})}_{cg} &= \sum_{i}\Tr(\rho Q_i) \frac{Q_i}{\Tr(Q_i)}=\sum_{i,j}\Tr(\rho Q_i)q(j|i)\frac{P_j}{\Tr(P_j)}\nonumber\\ &= \sum_{j}\left(\sum_{i}\Tr(\rho Q_i)q(j|i)\right)\frac{P_j}{\Tr(P_j)} = \sum_{j}\mu_j \frac{P_j}{\Tr(P_j)},
    \end{align}
    where we define the probability vector $\v \mu$ with elements
    \begin{equation}
        \mu_j = \sum_{i}\Tr(\rho Q_i)q(j|i).
    \end{equation}
    Substituting $Q_i$ from Eq. \eqref{Q_i} we obtain, 
    \begin{align}\label{bistocastisity}
        \mu_j &= \sum_{i}\Tr(\rho Q_i)q(j|i) = \sum_{i,m}D_{i,m}\Tr(\rho P_m)q(j|i) \nonumber\\&= \sum_{m} \left(\sum_{i}D_{i,m}q(j|i)\right)\Tr(\rho P_m):=\sum_{m}B_{j,m}\Tr(\rho P_m),
    \end{align}
    where
    \begin{equation}
        B_{j,m}=\sum_{i}D_{i,m}q(j|i).
    \end{equation}
    Now, note that 
    \begin{align}
        \forall m, \quad \sum_{j}B_{j,m} &= \sum_{i,j}D_{i,m}q(j|i) \nonumber\\&= \sum_{i}D_{i,m}\sum_{j}q(j|i) = \sum_{i}D_{i,m} =1.
    \end{align}
    On the other hand,
    \begin{align}\label{above_sum}
        \forall j,\quad \sum_{m}B_{j,m} = \sum_{i,m}D_{i,m}q(j|i) &= \sum_{i,m}D_{i,m}\times \frac{D_{i,j}\Tr(P_j)}{\sum_{k}D_{i,k}\Tr\left(P_k\right)}\nonumber\\&= \sum_{i}\left(\sum_{m}D_{i,m}\right)\times \frac{D_{i,j}\Tr(P_j)}{\sum_{k}D_{i,k}\Tr\left(P_k\right)}.
    \end{align}
    If $\v P$ is a fine-grained measurement, then it is rank-one (as follows fro definition \ref{Defn_fine_grained_mmt})
    \begin{equation}
        \forall i,\quad\Tr(P_i) = 1,
    \end{equation}
    which simplifies the above expression given in Eq. \eqref{above_sum} as follows:
    \begin{equation}
         \forall j,\quad \sum_{m}B_{j,m}= \sum_{i}\left(\sum_{m}D_{i,m}\right)\times \frac{D_{i,j}}{\sum_{k}D_{i,k}}=\sum_{i}D_{i,j} =1,
    \end{equation}
    which makes the matrix $B$ with entries $B_{j,m}$ bistochastic. Furthermore, if $\v P$ is a fine grained measurement, we can see from Eq. \eqref{cg_P} and Eq. \eqref{cg_Q} that the eigenvalues of $\rho^{(\textbf{P})}_{\text{cg}}$ and $\rho^{(\textbf{Q})}_{\text{cg}}$ is given by 
    $\{\mu_j\}_j$ and $\{p_j\}_j$ i.e
    \begin{equation}
        \lambda(\rho^{(\textbf{P})}_{\text{cg}})=\v p \quad\quad \lambda(\rho^{(\textbf{Q})}_{\text{cg}})=\v \mu,
    \end{equation}
    where $\v p$ and $\v \mu$ are probability vector with entries $p_j$ and $\mu_j$, respectively. Now we have shown in Eq. \eqref{bistocastisity} that 
    \begin{equation}
         \lambda(\rho^{(\textbf{Q})}_{\text{cg}})=\v \mu=B\v p = B \lambda(\rho^{(\textbf{P})}_{\text{cg}}),
    \end{equation}
    where $B$ is a bistochastic matrix. Therefore, using Eq. \eqref{thm_majorization_equvt} 
    \begin{equation}
        \lambda(\rho^{(\textbf{P})}_{\text{cg}})\succ \lambda(\rho^{(\textbf{Q})}_{\text{cg}}),
    \end{equation}
    which proves the claim.
\end{proof}

\begin{thm}[Schur concavity of passive energy.]\label{SCofOE}
    Consider two states $\rho_1$ and $\rho_2$ described by a Hamiltonian $H = \sum_{i=1}^d E_i|E_i\rangle\langle E_i|$ with $E_{i}< E_{i+1}$ for all $i\in{1,\ldots,d-1}$. If $\lambda(\rho_1)\succ \lambda(\rho_2)$, then their passive energy obeys the following inequality:
    \begin{equation}
        \Tr(H\Pi_1)\leq\Tr(H\Pi_2),
    \end{equation}
    where $\Pi_1$ and $\Pi_2$ are passive version of the state $\rho_1$ and $\rho_2$ respectively as defined in Eq. \eqref{passive_state_energy}. Here $\lambda(\cdot)$ is the column vector constructed from the eigenvalues of matrix $\cdot$. 
\end{thm}
\begin{proof}
    Let us denote the vector constructed from the eigenvalues of $H$ as $\lambda(H)$ i.e.,
    \begin{equation}
        \lambda(H) = (E_1\; E_2\;\ldots\; E_d)^T \quad\text{where}\quad  \forall i\;\; E_{i}< E_{i+1}.
    \end{equation}
    We can write the following from $\lambda(\rho_1)\succ \lambda(\rho_2)$:
    \begin{eqnarray}
        &&\lambda(\rho_1)\succ \lambda(\rho_2)\Rightarrow \lambda(\rho_1)\succ \lambda^{\downarrow}(\rho_2) \nonumber\\&\Rightarrow& B\lambda(\rho_1)= \lambda^{\downarrow}(\rho_2) \nonumber\\&\Rightarrow& \lambda(H)^T B\lambda(\rho_1)= \lambda(H)^T\lambda^{\downarrow}(\rho_2) \nonumber\\&\Rightarrow& \lambda(H)^T \lambda^{\downarrow}(\rho_1)  \leq\lambda(H)^T B\lambda(\rho_1)= \lambda(H)^T\lambda^{\downarrow}(\rho_2)\nonumber\\
        &\Rightarrow& \Tr(H\Pi_1)\leq\Tr(H\Pi_2),
    \end{eqnarray}
    where we identify $\lambda(H)^T \lambda^{\downarrow}(\rho_2)=\Tr(H\Pi_1)$ and $\lambda(H)^T \lambda^{\downarrow}(\rho_1)=\Tr(H\Pi_2)$ as the passive energy of $\rho_1$ and $\rho_2$ from Eq. \eqref{pass_energy_major} to write the final $\Rightarrow$. To write the second $`\Rightarrow'$, we use the fact that if $\textbf{p}\succ\textbf{q}$, there exist a bistochastic matrix $B$ that connects the probability vector $\textbf{p}$ with $\textbf{q}$ in the following manner (See theorem \ref{thm_majorization_equvt})
    \begin{equation}
        B\textbf{p}=\textbf{q}.
    \end{equation}
    Here $\lambda^{\downarrow}(\cdot)$ denotes the eigenvalue of $\cdot$ arranged in non-increasing order i.e.,
    \begin{equation}
        \lambda^{\downarrow}(\cdot)_{i}\geq \lambda^{\downarrow}(\cdot)_{i+1} \quad \forall i.
    \end{equation}
    To write the inequality in the fourth $`\Rightarrow'$, we proceed by decomposing the Bistochastic matrix $B$ as a convex sum of permutations $\pi_j\in\text{Permutations}$:
    \begin{equation}
        B=\sum_{j}\alpha_j\pi_j\quad\text{where}\quad\forall j, \alpha_j\geq 0 \;\;\text{and}\;\; \sum_j\alpha_j=1.
    \end{equation}
   Since the eigenvalues of $H$ are arranged in increasing order, the following equality holds for any positive semidefinite matrix $A$:
    \begin{equation}
        \min_{\pi\in\text{Permutations}} \lambda(H)^T\pi\lambda(A) = \lambda(H)^T\lambda^{\downarrow}(A),
    \end{equation}
    as all the eigenvalues of $A$ are decreasingly ordered in the vector $\lambda^{\downarrow}(A)$. This observation allows us to write
    \begin{align}
        \lambda(H)^T \lambda^{\downarrow}(\rho_2)&= \lambda(H)^TB \lambda(\rho_1)\nonumber\\&=\sum_{j}\alpha_j\lambda(H)^T\pi_j\lambda(\rho_1)&\geq \sum_{j}\alpha_j\lambda(H)^T\lambda^{\downarrow}(\rho_1)\nonumber\\&=\lambda(H)^T\lambda^{\downarrow}(\rho_1).
    \end{align}

\end{proof}

Now it is straight-forward to prove theorem \ref{fmt}. Employling lemma \ref{fmt_lem}, we can say that if the coarse-grained measurement $\textbf{Q}$ is formed via classical post-processing of the fine-grained measurement $\textbf{P}$, then 
\begin{equation}
        \lambda(\rho_{\text{cg}}^{(\textbf{P})})\succ \lambda(\rho_{\text{cg}}^{(\textbf{Q})}).
    \end{equation}
Now employing theorem \ref{SCofOE}, we obtained 
\begin{equation}
    \Tr (H\Pi^{(\textbf{P})}_{\text{cg}}) \leq \Tr (H\Pi^{(\textbf{Q})}_{\text{cg}}),
\end{equation}
where $\Pi^{(\textbf{P})}_{\text{cg}}$ and $\Pi^{(\textbf{Q})}_{\text{cg}}$  are the passive state obtained from the coarse-grained states $\rho^{(\textbf{P})}_{\text{cg}}$ and $\rho^{(\textbf{Q})}_{\text{cg}}$ respectively. Therefore, we have the desired relation between observational ergotropy
\begin{align}
    R(\rho,\textbf{P})&=\Tr(H\rho)-\Tr (H\Pi^{(\textbf{P})}_{\text{cg}})\nonumber\\&\geq \Tr(H\rho)-\Tr (H\Pi^{(\textbf{Q})}_{\text{cg}})=R(\rho,\textbf{Q}).
\end{align}
\end{document}